\documentclass[a4paper,UKenglish,cleveref, autoref, thm-restate]{lipics-v2021}

\usepackage{graphicx}
\usepackage{mathtools}
\usepackage{amsmath}
\usepackage{amsthm}
\usepackage{amssymb}
\usepackage{appendix}
\usepackage[dvipsnames]{xcolor}
\usepackage{caption}

\usepackage[linesnumbered,ruled]{algorithm2e}
\usepackage{tabularx}

\newcommand{\indeg}{\ensuremath{\mathsf{indeg}}\xspace}
\newcommand{\outdeg}{\ensuremath{\mathsf{outdeg}}\xspace}
\newcommand{\ILCS}{\ensuremath{\mathsf{ILCS}}\xspace}
\newcommand{\ELCS}{\ensuremath{\mathsf{ELCS}}\xspace}
\newcommand{\LCS}{\ensuremath{\mathsf{LCS}}\xspace}
\newcommand{\share}{\ensuremath{\mathsf{share}}\xspace}

\newcommand{\pred}{\textrm{pred}}

\definecolor{darkorange}{rgb}{0.8, 0.4, 0.0}

\usepackage{tikz}
\usetikzlibrary{arrows, quotes}
\tikzstyle{Red Dot}=[fill=red, draw=black, shape=circle]
\tikzstyle{Green Dot}=[fill=white, draw=green, shape=circle]
\tikzstyle{Empty}=[fill=white, draw=black, shape=circle]
\tikzstyle{Black}=[fill=black, draw=black, shape=circle]
\tikzstyle{Red fill}=[-, fill={rgb,255: red,255; green,88; blue,91}]
\tikzstyle{grey fill}=[fill={rgb,255: red,207; green,207; blue,207}, draw=black, shape=circle]
\tikzstyle{Red Circle}=[fill=white, draw=red, shape=circle]

\tikzstyle{Left}=[<-]
\tikzstyle{Right}=[->]

\pgfdeclarelayer{nodelayer}
\pgfdeclarelayer{edgelayer}
\pgfsetlayers{edgelayer,nodelayer,main}
\usepackage{makecell}
\usepackage{cellspace} 
\usepackage{thmtools,thm-restate}

\title{Computing $k$-mers in Graphs}

\author{Jarno N. Alanko}{Department of Computer Science, University of Helsinki, Helsinki, Finland}{}{}{}
\author{Máximo Pérez-López}{Department of Computer Science, Technical University of Denmark, Copenhagen, Denmark}{}{}{}

\authorrunning{J. N. Alanko and M. Pérez-López} %TODO mandatory. First: Use abbreviated first/middle names. Second (only in severe cases): Use first author plus 'et al.'

\Copyright{Jarno N. Alanko and Máximo Pérez-López} %TODO mandatory, please use full first names. LIPIcs license is "CC-BY";  http://creativecommons.org/licenses/by/3.0/


 %TODO mandatory; please add comma-separated list of keywords

\category{} %optional, e.g. invited paper

\relatedversion{} %optional, e.g. full version hosted on arXiv, HAL, or other respository/website
\nolinenumbers %uncomment to disable line numbering

\EventEditors{John Q. Open and Joan R. Access}
\EventNoEds{2}
\EventLongTitle{42nd Conference on Very Important Topics (CVIT 2016)}
\EventShortTitle{CVIT 2016}
\EventAcronym{CVIT}
\EventYear{2016}
\EventDate{December 24--27, 2016}
\EventLocation{Little Whinging, United Kingdom}
\EventLogo{}
\SeriesVolume{42}
\ArticleNo{23}
\begin{document}

\maketitle

\begin{abstract}
We initiate the study of computational problems on $k$-mers (strings of length $k$) in labeled \emph{graphs}. As a starting point, we consider the problem of counting the number of distinct $k$-mers found on the walks of a graph. We establish that this is $\#$P-hard, even on connected deterministic DAGs. However, in the class of deterministic Wheeler graphs (Gagie, Manzini, and Sir{\'e}n, TCS 2017), we show that distinct $k$-mers of such a graph $W=(V, E)$ can be counted using $O(|W|k)$ or $O(n^4 \log k)$ arithmetic operations, where $n=|V|$, $m=|E|$ and $|W|=n+m$. The latter result uses a new generalization of the technique of prefix doubling to Wheeler graphs. To generalize our results beyond Wheeler graphs, we discuss ways to transform a graph into a Wheeler graph in a manner that preserves the $k$-mers.

As an application of our $k$-mer counting algorithms, we construct a representation of the de Bruijn graph of the $k$-mers that occupies $O(n_k + |W|k \log(\max_{1 \leq \ell \leq k} n_\ell) + \sigma\log m)$ bits of space, where $n_\ell$ is the number of distinct $\ell$-mers in the Wheeler graph, and $\sigma$ is the size of the alphabet. We show how to construct it in the same time complexity. Given that the Wheeler graph can be exponentially smaller than the de Bruijn graph, for large $k$ this provides a theoretical improvement over previous de Bruijn graph construction methods from graphs, which must spend $\Omega(k)$ time per $k$-mer in the graph.
\end{abstract}

%STACS sumbission guidelines state that the title page should contain title + abstract, and the first section should start on the next page.
\newpage
\setcounter{page}{1}

\section{Introduction}

% \jarno{Motivation from genome graphs. Brief literature review, point out that our problem is very under-studied. Some applications (e.g. Jaccard distances of $k$-spectra of graphs, randomly sampling of $k$-mers).}

In recent years, a sizeable body of literature in the field of computational biology has emerged on storing and indexing sets of $k$-mers (strings of length $k$). Typically, these $k$-mers are all the overlapping windows of length $k$ in some genome sequence. In strings, the problem of counting the number of distinct $k$-mers is easily solvable in linear time using a combination of suffix and LCP arrays~\cite{manber1993suffix, KasaiLCP}, which are computable in $O(n)$ time for polynomial alphabets~\cite{KarkkainenST}. There exist highly optimized practical tools for counting short $k$-mers on strings and encoding them efficiently by exploiting this overlap structure~\cite{marchet2021data}. 

However, a significant effort is underway to move from linear reference genomes into a more rich graph-based reference genome that encodes all variation in a species. Very little attention has been paid to the problem of counting $k$-mers in more general structures such as trees, DAGs and general graphs, some of which can encode an \emph{exponential} number of distinct $k$-mers. There is some practical work toward this direction in the field of pangenomics~\cite{siren2017indexing}, but the theoretical limits and possibilities remain largely unexplored. In this paper, we take the first steps towards mapping this uncharted territory.

First, we show that it is \#P-hard to compute the number of distinct $k$-mers in directed labeled graphs, even if we restrict ourselves to connected deterministic graphs without directed cycles. However, we find that in the special case of deterministic \textit{Wheeler graphs} \cite{gagie2017wheeler}, we can compute the number of $k$-mers in $O(|W|k)$ time, where $|W|$ denotes the total number of vertices and edges of $W$, with a simple dynamic programming approach. Pushing further, we study what it takes to improve the time complexity with respect to $k$. Here, the problem turns out to be significantly more complicated than on strings, but we manage an $O(n^4 \log k)$ algorithm by generalizing the technique of \textit{prefix doubling}~\cite{manber1993suffix} to Wheeler graphs. We believe there is room to improve the $O(n^4)$ factor. These results pave the way toward counting $k$-mers in arbitrary graphs without explicit enumeration, by reducing them to Wheeler graphs that preserve the $k$-mer set. 

As an application of $k$-mer counting, we present an algorithm that builds the node-centric de Bruijn graph of the $k$-mers of a Wheeler graph $W$ in $O(n_k + |W|k\log(\max_{1\leq\ell\leq k} n_\ell) + \sigma\log m)$ time, where $n_\ell$ is the number of $\ell$-mers in $W$. Crucially, the multiplicative factor of $k$ is on the size of the Wheeler graph, rather than the size of the potentially exponential de Bruijn graph (which contains one node per $k$-mer). Previous algorithms to build de Bruijn graphs from graphs or sets of $k$-mers spend $\Omega(k)$ time per $k$-mer inserted in the graph, which inevitably incurs a cost of $O(kn_{k})$. In the bioinformatics literature, we find the GCSA2 construction algorithm \cite{siren2017indexing}, which builds the de Bruijn graph of a graph by building all distinct $k$-mers of the graph by prefix doubling. Other de Bruijn graph construction algorithms, such as dynamic versions of the  data structure by Bowe et al. \cite{BOSSdBg, SuccintDynamicDBgsAlipanahi}, do not take a graph as input, and only allow the insertion of an explicit $k$-mer at a time, at a worst-case cost of $\Omega(k)$ per insertion. This application showcases how our theoretical results translate into practical gains for graph algorithms.

The remainder of this paper is organized as follows. In Section \ref{sec:prelim}, we lay down some basic notation and concepts. In Section~\ref{sec:complexity}, we discuss the computational complexity of the problem, and prove the problem to be $\#$P-hard. In Section~\ref{sec:kmer_counting}, we solve the problem on Wheeler graphs in time $O(|W|k)$ or $O(n^4 \log k)$. In Section~\ref{sec:k_wheelerization}, we discuss ways of turning an arbitrary graph into a Wheeler graph with the same $k$-mers in order to apply the algorithms in Section \ref{sec:kmer_counting}.  Section \ref{sec:dbg} describes the de Bruijn graph construction algorithm. Section~\ref{sec:conclusion} concludes the paper with some closing remarks.

\section{Notation and preliminaries}\label{sec:prelim}

We consider directed graphs with single character labels on the edges. Formally, we consider graphs $G$ formed by a vertex set $V(G)$ and an edge set $E(G)\subseteq V\times V\times \Sigma$, where $\Sigma$ is an alphabet of symbols that is totally ordered by a relation $\prec$. The indegree and outdegree of a node $v$ are denoted with $\indeg(v)$ and $\outdeg(v)$, respectively. A $k$-mer $\alpha$ is a string of length $k$ from $\Sigma^k$, indexed from $1$ to $k$. We compare strings by extending $\prec$ to $\Sigma^*$ co-lexicographically, which is the right-to-left lexicographical order (we abbreviate this by \textit{colex} order). We say that a labeled graph $G$ contains a $k$-mer $\alpha$ if there exists a walk $v_0, v_1, \ldots, v_k$ in $G$ such that $(v_{i-1}, v_i, \alpha[i])\in E(G)$ for all $1\leq i \leq k$. In this case, we say that $\alpha$ \textit{reaches} $v_k$, and write $\alpha\leadsto v_k$. The minimum and maximum $k$-mer that reach a vertex $v$ are denoted $\min^k_{v}$ and $\max^k_{v}$, respectively. A Wheeler graph is a labeled directed graph with the following properties \cite{gagie2017wheeler}:

\begin{definition} \label{def:wheeler_graph}
A labeled directed graph $G$ is a \textbf{Wheeler graph} if there is an ordering on the nodes such that nodes with no incoming edges precede those with positive indegree, and every pair of edges $(u,v,a), (u',v',a')$ of the state graph satisfies the following two properties:
\begin{itemize}
    \item W1: $a \prec a' \Rightarrow v < v'$
    \item W2: $(a = a' \land u < u') \Rightarrow v \leq v'$
\end{itemize}
\end{definition}
We say that a graph is \textit{deterministic} if no two outgoing edges from the same vertex have the same label. In deterministic Wheeler graphs, we have the following stronger property:
\begin{lemma}\label{lemma:det_wheeler_graph} In deterministic Wheeler graphs, for any pair of edges, $(u,v,a), (u',v',a')$ we have $(a = a' \wedge v<v') \Rightarrow u < u'$. $\qed$
\end{lemma}
There exists an algorithm to transform a non-deterministic Wheeler graph into a deterministic one with the same $k$-mers, that in the worst case is only a linear factor larger \cite{alanko2020regular}. In contrast, in general graphs the deterministic version could be exponentially larger. However, the running time of the algorithm is $O(n^3)$ \cite{alanko2020regular}, and improving this upper bound is an open problem. Thus, we assume that our input are deterministic graphs throughout the paper.
%\noindent A nondeterministic Wheeler automaton is called an WNFA, and a deterministic one is called a WDFA.

\noindent In a labeled graph we can identify the incoming \textit{infimum} and \textit{supremum} strings to every vertex, denoted $\inf_v$ and $\sup_v$ for all $v\in V$. We use the definition of \cite{DBLP:conf/cpm/AlankoCCKMP24}, though note that they read the strings backwards, and use lexicographical order, while we read the strings forward and use colex order instead. Let $\gamma_1,\gamma_2,\ldots,\gamma_{2n}$ be the colexicographically sorted infimum and supremum strings of the $n$ vertices of a Wheeler graph. The \textit{longest common suffix} array is defined as an array $\mathsf{LCS}[2..2n]$, where $\mathsf{LCS}[i]$ is the longest common suffix of $\gamma_{i-1}$ and $\gamma_i$. It was first defined in \cite{DBLP:conf/dcc/ConteCGMPS23}, where it was also proven that in a Wheeler graph, $\inf_v\preceq\sup_v$, and $\sup_u\preceq\inf_v$ for all $u<v$ in Wheeler order. Observe that we can find the \LCS value between any infima and suprema in $O(1)$ time, by building a range minimum query data structure \cite{Fischer-RMQ, original_rmq} over the \LCS array, and performing the appropriate query.

We find it convenient to divide the \LCS array into two parts: one is the \textit{external} \LCS array, written \ELCS, that holds the \LCS between the supremum and infimum string of $v_{i-1}$ and $v_{i}$. The other one is the \textit{internal} \LCS array, written \ILCS, that holds the \LCS value of the infimum and supremum of the same vertex. We compute them as $\ELCS(v_{i-1}, v_i)=\LCS[2i-1]$, and $\ILCS(v_i)=\LCS[2i]$. In Section \ref{sec:kmer_counting} we show how to use these arrays in the context of $k$-mers.

\textbf{Model of computation}. Since the number of $k$-mers in a graph may be exponential, the time complexity of integer addition and multiplication of large integers can become relevant for the running time of our algorithms. We put this issue to the side by assuming a Word RAM model \cite{Hagerup1998} with a word size large enough to hold the maximum numbers that we handle, which are proportional to the maximum number of $\ell$-mers in the graph, for any $1\leq\ell\leq k$.
%Even on real-world machines, our counting algorithms can be run within the stated time bounds to compute the number of $k$-mers modulo some integer $m$ by running the arithmetic in the ring $\mathbb Z_m$.

\section{Computational complexity of counting $k$-mers in graphs}\label{sec:complexity}

%\begin{figure}[t]
%    \centering
%    \includegraphics[width=1.0\linewidth]{figures/reduction.png}
%    \caption{Reduction from $\#DNF$ to $k$-mer counting on a deterministic connected acyclic graph.}
%    \label{fig:DNF-reduction}
%\end{figure}

\begin{figure}
    \centering
    \includegraphics[width=\linewidth]{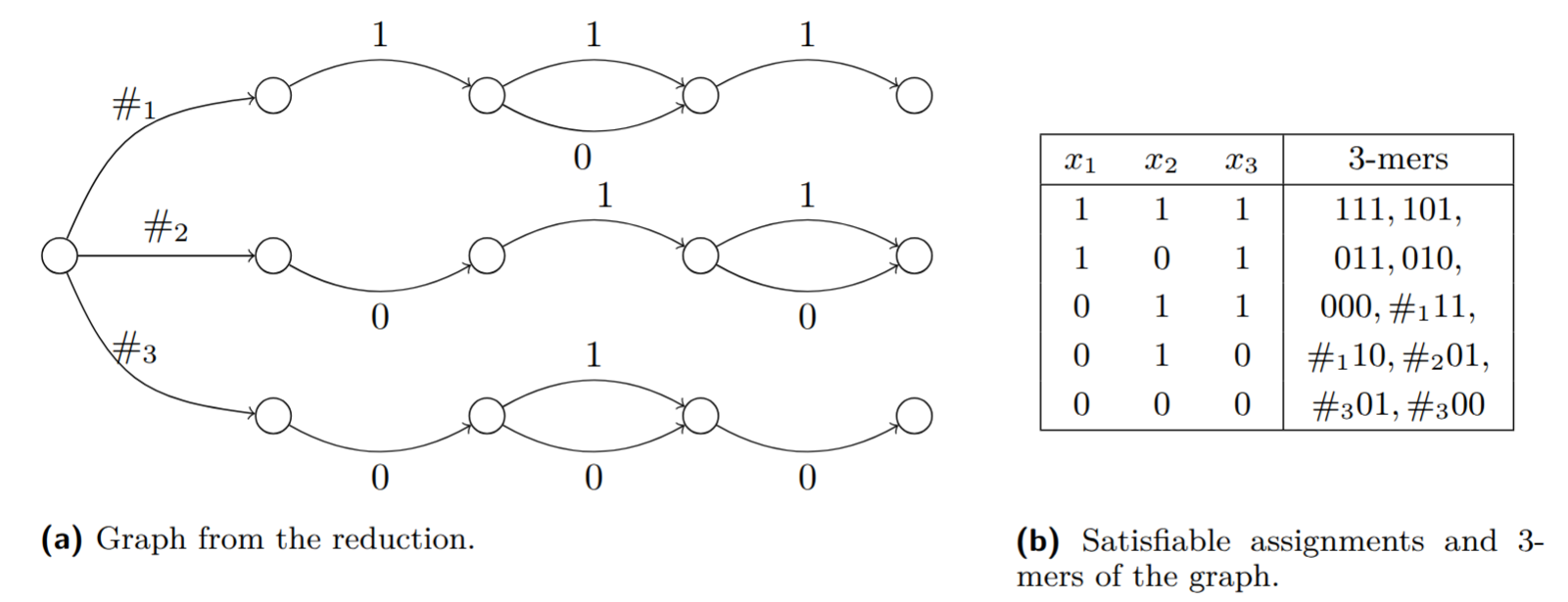}
    \caption{Reduction from the example DNF formula $F=(x_1\wedge x_3) \vee(\neg x_1\wedge x_2)\vee(\neg x_1\wedge \neg x_3)$, to counting $3$-mers. Note how $d_1=1, d_2=0, d_3=1$ and thus $\#SAT(F)=10-2^1-2^0-2^1=5$.}
    \label{fig:DNF-reduction}
\end{figure}

In this section we show that the $k$-mer counting problem on graphs is $\#$P-hard even when the graph is acyclic, connected (i.e., there exists a node from which every other node can be reached) and deterministic (i.e., no two edges with the same label have the same origin).

Let $\#SAT(F)$ denote the number of satisfying assignments to a boolean formula $F$ over $n$ variables $x_1, \ldots, x_n$ in conjunctive normal form. Computing $\#SAT(F)$ is the archetypical $\#P$-complete problem \cite{valiant1979complexity}. It is also well known that the problem is $\#P$-complete in disjunctive normal form since we can write the number of solutions to a CNF formula $F$ as $2^n - \#SAT(\neg F)$, where $\neg F$ can be written in disjunctive normal form of the same size as $F$ by de Morgan's laws.
\begin{theorem}\label{thm:sharp_P}
    Counting the number of distinct $k$-mers on walks on an acyclic deterministic connected graph $G$ is $\#$P-hard.
\end{theorem}
\begin{proof}
We reduce from counting satisfying assignments of a DNF formula. Let $n$ be the number of distinct variables in the formula. For each clause, we introduce a graph gadget. The gadget for clause $C_j$ is as follows: add one node $v_i$ for each $i = 0, \ldots, n$, and edges $(v_{i-1},v_{i},c)$, for all variables $x_i$ that appear in $C_j$, such that the label $c$ is 1 if $x_i$ appears as positive in $C_j$, otherwise 0. For variables $x_i$ that do not appear in $C_j$, we add both edges $(v_{i-1}, v_i, 0)$ and $(v_{i-1}, v_i, 1)$. Now each distinct $n$-mer in this gadget corresponds to a satisfying truth assignment to the clause, such that the $i$-th character of the $n$-mer is 1 iff variable $x_i$ is set to true. Let $G$ be the union of all the gadgets for all clauses in the input.
Now, the number of distinct $n$-mers in $G$ is equal to the number of satisfying truth assignments. To make the graph connected, we add a new source node that has an edge to the sources of each of the gadgets, with a unique label $\#_j$ for each gadget $C_j$. See Figure \ref{fig:DNF-reduction}. Let us denote the modified graph with $G_\#$. The addition of the new source node adds $2^{d_j}$ distinct $n$-mers ending in the gadget for $C_j$, where $d_j$ is the number of nodes with two outgoing edges among the first $(n-1)$ nodes of the gadget. Thus, if the number of $n$-mers in $G_\#$ is $N$, and the number of clauses is $m$, then the number of satisfying assignments to the DNF formula is $N - \sum_{j=1}^m 2^{d_j}$.
\end{proof}
\noindent It is $\#P$-hard even to count the number of $k$-mers arriving at a single node:
\begin{corollary}
Counting the number of $k$-mers arriving at a single node $v$ in a graph $G$ is $\#$P-hard.
\begin{proof}
If there were an efficient algorithm to count the $k$-mers coming into a single node $v$, then we could add a new sink state to any graph, with edges labeled with a unique symbol $\#$ pointing to the sink: now the number of $(k+1)$-mers ending at the sink equals the number of $k$-mers in the original graph.
\end{proof}
\end{corollary}
\noindent Theorem \ref{thm:sharp_P} also implies that counting factors of length $k$ in a regular language defined by a DFA is $\#$P-hard:
\begin{corollary}
Counting factors of length $k$ in the language of a DFA is $\#$P-hard.
\end{corollary}
\begin{proof}
We can turn a connected deterministic graph $G$ into a DFA by picking a node that can reach all states as the initial state, and making every state accepting. Now counting factors of length $k$ in the language of the automaton counts $k$-mers in the graph.
\end{proof}

\section{Algorithms for $k$-mer counting in deterministic Wheeler graphs}\label{sec:kmer_counting}

In this section, we develop two polynomial-time algorithms for counting $k$-mers on deterministic Wheeler graphs. We abbreviate those by \textit{DWGs}. Let $W=(V, E)$ be a DWG, with $|V|=n$, $|E|=m$ and $|W|=n+m$. The first algorithm has time complexity $O(|W|k)$. For the second algorithm, we reduce the dependency on $k$, but we end up with a time complexity of $O(n^4 \log k)$.

These algorithms require, as an internal step, to obtain the LCS array up to a length of $k$, i.e. $\min\{\mathsf{LCS}(\gamma_i, \gamma_{i+1}), k\}$ for all values in the array. The current known \textsf{LCS} construction algorithms take $O(n\log \sigma + \min\{m\log n, m + n^2\})$ time on general deterministic graphs \cite{DBLP:conf/cpm/AlankoCCKMP24}, and $O(n\log\sigma + m)$ if the graph is a Wheeler \emph{semi-DFA}, which requires to have exactly one source. We do not make the assumption of a single source, thus we can only use the former algorithm, which is within the time bound for our $O(n^4\log k)$ counting algorithm. For our $O(|W|k)$ counting algorithm, we present a tailored \textsf{LCS} construction that is constructed during the $k$-mer counting algorithm at no extra cost, but with values capped to $k$, which is enough for our purposes.

\subsection{Counting $k$-mers in $O(|W|k)$ time}\label{sec:counting_kmers_Wk}
Let $S_v^k$ be the set of distinct $k$-mers reaching a vertex $v$. The algorithm relies on two properties of deterministic Wheeler graphs. First, the incoming edges to a node all have the same label (a direct consequence of the first condition in Definition \ref{def:wheeler_graph}), a property called \textit{input-consistency}. We denote with $\lambda(v)$ the label on the incoming edges to $v$, setting $\lambda(s)=\varepsilon$ for any source $s$. Second, we have the following well-known property on the colexicographic order of incoming $k$-mers, which we prove in appendix \ref{appendix:lemma6}:
\begin{restatable}{lemsix}{lemmashared}\label{lemma:kmer_order}
For two nodes $u,v$ in a DWG, if $u < v$, then for any two $k$-mers $\alpha$ and $\beta$ such that $\alpha \leadsto u$ and $\beta \leadsto v$, we have $\alpha \preceq \beta$.
\end{restatable}

\begin{corollary}\label{corollary:max_min}
    If two different nodes $u < v$ of the DWG share a $k$-mer, then they share $\max_{u}^k=\min_v^k$. $\qed$
\end{corollary}

%\begin{corollary}\label{corollary:max_min}
%If a graph $G$ is a Wheeler DFA, then for any pair of nodes $u < v$, it holds that $\max S_u^\ell \leq \min S_v^\ell$ for all $\ell \leq k$. %$\qed$
%\end{corollary}
If the DWG has sources, some vertices may not have incoming $k$-mers, but they always have an \textit{infimum} and a \textit{supremum} string, as defined in \cite{DBLP:conf/cpm/AlankoCCKMP24}. We also prove that, if two vertices $u<v$ share a $k$-mer, then $\max^k_u$ and $\min^k_v$ are suffixes of $\sup_u$ and $\inf_v$, respectively. If they do not share a $k$-mer, that may not hold in general. This makes it possible to check for shared $k$-mers using the original \textsf{LCS} array.
\begin{restatable}{lem}{lemmaorder}\label{lemma:lcs_and_kmers}
    Two different nodes $u < v$ of a DWG share a $k$-mer if and only if $\mathsf{LCS}(\sup_u,\inf_v)\geq k$. In that case, $\max^k_u$ is a suffix of $\sup_u$, and $\min^k_v$ is a suffix of $\inf_v$.
\end{restatable}
We leave the proof in appendix \ref{appendix:lemma6}. Now, we show a recurrence to build the sets $S_v^\ell$ for all $\ell = 1,\ldots,k$. Let $\pred(v)$ be the set of nodes that have an outgoing edge to $v$. Then, for a non-source vertex, we have
\begin{equation}\label{eq:union_formula}
S_v^\ell = \bigcup_{u \in \pred(v)} \{ \alpha \cdot \lambda(v) \mid \alpha \in S_u^{\ell-1} \},\end{equation}
with the base case $S_v^0 = \{\varepsilon\}$ for all $v \in V$, where $\varepsilon$ denotes the empty string. For sources, $S_v^\ell=\emptyset$ for $\ell>0$.

However, at this point, we are not interested in enumerating the distinct $\ell$-mers, only counting them. Let $C_\ell(v) = |S_v^\ell|$. Let $u_1, \ldots, u_d$ be the in-neighbours of $v$ sorted in the Wheeler order. For $\ell=0$, we have $C_\ell(v)=1$ for all $v$. Then, for $\ell>0$, sources have $C_\ell(s)=0$; and for non-source vertices, the size of the union in Equation \ref{eq:union_formula} can be written as:
\[C_\ell(v) = \sum_{i = 1}^d C_{\ell - 1}(u_i) - \sum_{i=2}^d |S_{u_{i-1}}^{\ell-1} \cap S_{u_i}^{\ell-1}| = \sum_{i = 1}^d C_{\ell - 1}(u_i) - \sum_{i=2}^d [\max S_{u_{i-1}}^{\ell-1} = \min S_{u_i}^{\ell-1}], \]
where $[P]$ is the Iverson bracket that evaluates to 1 if the predicate $P$ is true, otherwise to 0. The last equality holds thanks to Corollary \ref{corollary:max_min}. Now the only difficult part in computing the recurrence is determining whether $\max S_{u_{i-1}}^{\ell-1} = \min S_{u_i}^{\ell-1}$. As discussed before, if we have the LCS array available, by Lemma \ref{lemma:lcs_and_kmers}, we can directly check it in the array. Otherwise, there exists a quickly computable recurrence for that as well.
Let $\ell\geq 1$. For non-source nodes $u, v$, let $\mathsf{share}_\ell(u,v) = [\max S_{u}^{\ell} = \min S_{v}^{\ell}]$, and let $\mathsf{share}_\ell(u,v)=0$ if $u$ or $v$ are sources. Furthermore, for non-source $v$ define \begin{equation*}\ILCS_{\geq \ell}(v)=[|\text{inf}_v|\geq\ell]\cdot [|\text{sup}_v|\geq\ell]\cdot[\min S^\ell_v = \max S^\ell_v],\end{equation*} and $\ILCS_{\geq\ell}(s)=0$ for a source $s$. We have that $\mathsf{share}_\ell(u,v) = 1$ iff $\lambda(u) = \lambda(v)$ and $\mathsf{share}_{\ell-1}(\max \pred(u), \min \pred(v)) = 1$, with the base case $\share_{0}(u,v) = 1$ for all $u,v$. Similarly, we also have that $\ILCS_{\geq\ell}(v)=1$ iff $\ILCS_{\geq\ell-1}(\min \pred(v))=1$, $\ILCS_{\geq\ell-1}(\max\pred(v))=1$ and $\share_{\ell-1}(\min\pred(v),\max\pred(v))=1$, with the base case $\ILCS_{\geq 0}(v)=1$ for all $v$. To compute these without tabulating all $O(|V|^2)$ input pairs $u,v$ to $\share_\ell$, we exploit the formula 
\begin{align}\share_\ell(u,v) = \min\{\min_{u<w\leq v} \share_\ell(w-1, w), \min_{u<w<v} \ILCS_{\geq\ell}(w)\}.\label{equation:share}\end{align}
With this, it suffices to compute $\share_\ell$ only for adjacent inputs $(v_{i-1}, v)$, and build a range minimum data structure \cite{Fischer-RMQ,original_rmq} over both the $\ILCS_{\geq \ell}(v)$ and $\share_\ell(v_{i-1}, v_i)$ arrays, to be able to resolve $\share_\ell(u,v)$ in constant time for any input pair. Proceeding level by level $\ell = 1,2,\ldots,k$, we obtain:

\begin{theorem}\label{theorem:counting_dp}
   We can compute the number of distinct $k$-mers in the state graph $(V,E)$ of a DWG in total time $O((|V| + |E|)k)=O(|W|k)$. \qed
\end{theorem}
Note that during the procedure to compute $\mathsf{share}_\ell$ and $\ILCS_{\geq\ell}$ for all $\ell$, we can also infer the values of the \ILCS and \ELCS array up to $k$, and keep all needed information in $O(n)$ space.
\subsection{Counting $k$-mers in $O(n^4\log k)$ time in DWGs}
Now, we present a new algorithm for $k$-mer counting that has a logarithmic dependency on $k$, based on the technique of \emph{prefix-doubling} \cite{manber1993suffix}. We obtain a running time of $O(n^4\log k)$, which is better than our previous algorithm for $k=\omega(n^2\log n)$. We suspect that the upper bound could be improved, but we were not able to. We give now an intuition for the difficulty of this problem.

Our starting point is the algorithm by Kim, Olivares and Prezza \cite{kim_et_al:LIPIcs.CPM.2023.16}, that computes the total order of the infima of every vertex by prefix doubling. Let $suf_\ell(\inf_v)$ denote the suffix of length $\ell$ of $\inf_v$. For every vertex $v$ and $\ell=2^i$ with $i=1,\ldots,\log k$, their algorithm maintains a set $P_\ell(v)$ of the vertices $u$ such that $suf_\ell(\inf_v)$ starts at $u$ and ends at $v$. Then, it computes $P_{2\ell}(v)$ as the union of the $P_\ell(u)$ such that $u\in P_\ell(v)$, and $u$ has the smallest $suf_\ell(\inf_u)$ among all $u\in P_\ell(v)$. It observes that the mentioned union is disjoint because of determinism, and thus computing $P_{2\ell}(v)$ from $P_\ell(v)$ takes $O(n)$ time. Across all vertices and lengths up to $k$, the running time is $O(n^2\log k)$, given that a sorting step is done by radix sort in linear time. 

Our problem is significantly more complicated since we are interested not only in the infima, but \emph{all} the strings reaching a node. We assume that we already have a total order on the nodes, but the problem remains challenging. First, we decompose the problem into computing $C(v)$ for each $v \in V$. From these, it is easy to obtain the $k$-mer count of the whole graph because adjacent nodes in the Wheeler order share are most one $k$-mer, which we can detect using the \ILCS array like in the previous section. Further, we decompose computing $C_k(v)$ into counts $C_k(u,v)$ for all $u \in V$, where $C_k(u,v)$ denotes the number of $k$-mers starting from $u$ and ending at $v$. The counts $C_k(u,v)$ also turn out easy to compute by squaring the adjacency matrix repeatedly, but for $C_k(v)$ we must be very careful to avoid double counting between the $C_k(u,v)$ of different starting nodes $u$. Let $S_k(u,v)$ denote the $k$-mers starting from $u$ and ending at $v$. We have $C_k(v) = |\bigcup_{u \in V} S_k(u, v)|$. This reduces the problem into computing sizes of unions of the sets $S_k(u, v)$, or computing their intersections to subtract double counting. This is where the core difficulty lies because despite sharing the end point $v$, the shape of the intersection $S_k(u, v) \cap S_k(w, v)$ can be complicated and arbitrarily large.

%In our case, we have the total order of the vertices, but we need the count a possibly exponential number of $k$-mers reaching the vertices, not just the infima. Let $S_{\ell}(u, v)$ be the set of distinct $\ell$-mers that start at $u$ and end at $v$, and let $C_{\ell}(u, v)=|S_\ell(u,v)|$. We can compute all $C_\ell(u,v)$ easily by squaring the adjacency matrix repeatedly. We can now write $C_\ell(v) = |\bigcup_{u \in V} S_\ell(u, v)|$, but the size of the union is hard to compute because for any two distinct $u_a, u_b\in V$, $S_\ell(w_a, v)\cap S_\ell(w_b, v)$ may be arbitrarily large, not just $0$ or $1$ like when we are only computing infima.

Therefore, we devise a novel recursive method to compute unions of $S_\ell(u,v)$ for \emph{ranges} of $u$. Define $T_\ell([u_i,u_j], v)$ as $|\bigcup_{w\in [u_i,u_j]} S_{\ell}(w,v)|$ for any interval of vertices $[u_i,u_j]$ and any $v$. 
%Ideally, we would only need to compute $T_\ell([u_1,u_n],v)=C_\ell(v)$ for every vertex. % This is already said above in another way
If all the $(\ell/2)$-mers reaching two vertices $w_a,w_{a+1}$ are distinct, the number of $\ell$-mers reaching $v$ and passing through $w_a,w_{a+1}$ is $C_{\ell/2}(v,w_a)\cdot T_{\ell/2}([u_1,u_n],w_a)+C_{\ell/2}(v,w_{a+1})\cdot T_{\ell/2}([u_1,u_n],w_{a+1})$, because the first halves are distinct. However, if it happens that $w_a,w_{a+1}$ share one $(\ell/2)$-mer $\beta$, the number of $\ell$-mers reaching $v$ with prefix $\beta$ depends on the size of the union between $S_{\ell/2}(w_a,v)$ and $S_{\ell/2}(w_{a+1},v)$, and thus we recurse to length $\ell/2$. 
\begin{figure}
    \centering
    \begin{tabular}{c|c|c}
    \multirow{-4}{0.2\textwidth}{
    \resizebox{2.5cm}{3.0cm}{
    \begin{tikzpicture}
	\begin{pgfonlayer}{nodelayer}
		\node [style=Empty] (0) at (0, 1.75) {1};
		\node [style=Empty] (1) at (-1, 0.75) {2};
		\node [style=Empty] (2) at (1, 0.75) {3};
		\node [style=Empty] (3) at (-1, -0.5) {4};
		\node [style=Empty] (4) at (1, -0.5) {5};
		\node [style=Empty] (5) at (0, -1.5) {6};
	\end{pgfonlayer}
	\begin{pgfonlayer}{edgelayer}
		\draw [style=Right] (1) to node [left] {$a$} (0);
		\draw [style=Right] (2) to node [right] {$a$} (0);
		\draw [style=Right] (3) to node [left] {$a$} (1);
		\draw [style=Right] (4) to node [right] {$a$} (2);
		\draw [style=Right] (5) to node [left] {$a$} (3);
		\draw [style=Right] (5) to node [right] {$b$} (4);
		\draw [style=Right] (0) to node [right] {$c$} (5);
	\end{pgfonlayer}
\end{tikzpicture}}} & \multirow{-4.5}{0.3\textwidth}{\includegraphics[width=0.9\linewidth]{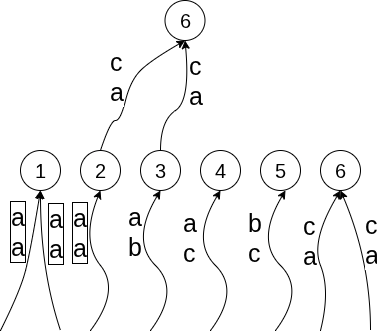}} & \makecell[l]{\begin{tabular}{|c|c|c|}
    \hline
    $w$ & $T^W_2([1,6],w)$ & $T_2([1,6],w)$\\
    \hline
    1 & 0 & 0 \\
    \hline
    2 & 0 & 0 \\
    \hline
    3 & 1 & 0 \\
    \hline
    4 & 1 & 1 \\
    \hline
    5 & 1 & 1 \\
    \hline
    6 & 1 & 1 \\
    \hline
\end{tabular} \\
$B_2=\{[1,2]\}, R([1,6],aa)=[1,2]$}
 \\
    \end{tabular}
    \caption{Example of a Wheeler graph, with an illustration of how we compute the $4$-mer counts of vertex $6$. The $2$-mers encapsulated with a box are black, and the rest are white. With the information on the table, we can calculate $T_4([1,6],6) = T^W_2([1,6],3)\cdot C_2(3,6)+T_2([1,2],6)=2$.}
    \label{fig:placeholder}
\end{figure}
\subsubsection{Recursive algorithm}
Because of input consistency, $C_1(v)=[\deg_{in}(v)>0]$. For $\ell>1$, we have $C_\ell(v)=T_\ell([u_1,u_n], v)$, and we give a recursive formulation for $T_{\ell}([u_i,u_j],v)$. First, we introduce some notation. Given $\ell\geq 1$, we call an $\ell$-mer \emph{white} if all its occurrences arrive at the same vertex. Otherwise, we call the $\ell$-mer \emph{black}. Let $T_\ell^W([u_i,u_j], v)$ be the number of white $\ell$-mers starting from any vertex in $[u_i,u_j]$ and reaching $v$, and let $B_\ell$ be the set of all black $\ell$-mers of length $\ell$. Finally, let $R([u_i,u_j], \beta)$ denote the maximal interval of vertices that we can reach starting from $[u_i,u_j]$ and following $\beta$. Note that for two distinct black $\ell$-mers $\beta_1\prec\beta_2$ and any $I\subseteq [u_1,u_n]$, $|R(I, \beta_1)\cap R(I, \beta_2)|\leq 1$.

We compute $T_\ell([u_i,u_j],v)$ by dividing the number of $\ell$-mers reaching $v$ from $[u_i,u_j]$ into those whose first half is white, and those whose first half is black:
\begin{gather}
T_\ell([u_i,u_j], v)=\sum_{w\in V}T^W_{\ell/2}([u_i,u_j],w)\cdot C_{\ell/2}(w,v)+\sum_{\beta\in B_{\ell/2}} T_{\ell/2}(R([u_i,u_j],\beta), v)\label{equation:t_recursion}
\end{gather}
The first term of the sum avoids double counting because all the first white halves are different. The second term counts the number of $\ell$-mers that reach $v$ with each black $(\ell/2)$-mer as a prefix, starting from $[u_i,u_j]$.
%In the section below we show how to compute $T^W_{\ell/2}([u_i,u_j],v)$ from $T_{\ell/2}([u_i,u_j],v)$ by subtracting the possible one or two black $(\ell/2)$-mers that may originate from $[u_i,u_j]$. We also show how to find the intervals $R([u_i, u_j], \beta)$ for $\beta\in B_{\ell/2}$. This part has been moved to the appendix

To fulfill the recursion, we run into the issue that for two $(\ell/2)$-mers $\beta$ and $\gamma$, $R(\beta)\subseteq R(\beta\gamma)$, and we cannot easily infer the value of $T_{\ell}(R(\beta\gamma),v)$ from $T_{\ell}(R(\beta),v)$. We settle for computing $T_{\ell}([u_i,u_j],v)$ for all $O(n^2)$ intervals $[u_i,u_j]\subseteq [u_1,u_n]$ at every level of the recursion, and for every $v\in V$. This requires a table of size $O(n^3)$ at every level of the recursion, where each value takes $O(n)$ time to compute with formula \ref{equation:t_recursion} once the terms inside the sum have been computed. 

Once we have found all $T_{\ell/2}([u_i,u_j], w)$, the terms $T^W_{\ell/2}([u_i,u_j], w)$ for the next recursion can be computed by subtracting the black $k$-mers reaching the vertices, which we can check from the sets $P_{\ell/2}(v)$ as shown in appendix \ref{appendix:n4logk}. There, we also show how to compute the intervals of $B_\ell$ and $R([u_i, u_j], \beta)$ from the \textsf{LCS} array and the sets $P_{\ell/2}(v)$. We achieve this by building a data structure of size $O(n^2)$ that allows us to compute each $T^W_\ell([u_i,u_j],v)$ from $T_\ell([u_i,u_j],v)$ in $O(1)$ time, and to find the intervals $B_\ell$ in $O(n)$ time, and all $R([u_i,u_j],\beta)$ in a total of $O(n^3)$ time.

With $\log k$ recursion steps, we arrive at an algorithm with $O(n^4\log k)$ running time, and $O(n^3)$ space. The space complexity can be achieved by only keeping in memory the table from the last iteration. 
All together, we have shown the following theorem:
\begin{theorem}
    We can compute the number of distinct $k$-mers of a DWG $W=(V, E)$ in time $O(n^4\log k)$, where $n=|V|$.
\end{theorem}

\section{$k$-mer preserving graph transformations to Wheeler graphs}\label{sec:k_wheelerization}
The algorithms in Section \ref{sec:kmer_counting} are of limited usability since they apply only to deterministic Wheeler graphs. In this Section we discuss how to transform some more general graphs into deterministic Wheeler graphs.

An \emph{layered DAG} with $k$ layers is a directed acyclic graph that can be partitioned into $k$ disjoint subsets of vertices $V_0, \ldots V_{k-1}$ such that edges exists only from vertices in $V_i$ into $V_{i+1}$ for $i = 0, \ldots, k-2$. For any graph, we can build the $k$ times unfolded DAG, which preserves the $k$-mers in the graph:
\begin{definition}\label{def:k_times_unfolded_DAG}
    The $k$ times unfolded DAG of a graph $G$ is a layered DAG $G_k$ with $k+1$ layers such that the vertices on each layer $V_i$ are a copy of the vertices of $G$, and the edge set is defined so that for every edge $(u,v,c)$ in $G$ there are edges from the copy of $u$ in layer $i$ to the copy of $v$ in layer $i+1$ for all $i = 0, \ldots k-1$.
\end{definition}
\noindent We can turn this into a DFA by introducing a new initial state with epsilon-transitions into the first layer, and determinizing with the classic subset construction. Now, since every $k$-mer corresponds to a distinct path from the initial state to the final layer, we can count $k$-mers by counting the distinct paths to each node on the final layer of the DAG, which is easily done with a dynamic programming algorithm that processes the nodes in topological order.

However, this does not exploit overlaps between $k$-mers. For deterministic acyclic graphs that have a source node that can reach all nodes, we may use the Wheelerization algorithm of Alanko et al. \cite{alanko2020regular} that computes the smallest equivalent Wheeler DFA in $O(n \log n)$ time in the size of the output. More generally, if the input is the state graph $G$ of any DFA recognizing a Wheeler language, the minimum equivalent Wheeler DFA can be constructed in $O(|G|^3 n \log n)$ time, where $n$ is the size of the output \cite{d2023ordering}.

Another option is to build the de Bruijn graph of the $k$-mers, which is always Wheeler \cite{gagie2017wheeler}. Unfortunately, this does not help much with counting the $k$-mers because building the de Bruijn graph takes time proportional to the number of $k$-mers, since there is one node per $k$-mer. That is, building the de Bruijn graph amounts to a brute force enumeration of all $k$-mers. However, the de Bruijn graph is not necessarily the smallest graph for a given set of $k$-mers: For example, the Wheeler DFA minimization algorithm of Alanko et al. \cite{alanko2022linear}, run on the de Bruijn graph, can reduce the size of the graph while maintaining Wheelerness. Building this compressed graph directly without first building the de Bruijn graph would be useful, but seems to be  nontrivial. We leave this open for future work.

\section{The de Bruijn Graph of a Deterministic Wheeler Graph}\label{sec:dbg}

In this section, we define a representation of the node-centric de Bruijn graph of the $k$-mers of a DWG $W$, building on the definitions of black and white $k$-mers and the counting algorithms introduced in the previous section.
We denote the de Bruijn graph by $dBg$. First, we describe the data structure contents, then we describe the algorithms to navigate the de Bruijn graph given this data structure, and then we analyse their time and space in bits. Finally we show how to build it in time and space $O(n_k + |W|k\log(\max_{1\leq\ell\leq k} n_\ell) + \sigma\log m)$.

\subsection{Data structure}
Before presenting the data structure components, we define some new concepts. 
\begin{itemize}
    \item A \emph{left-extension} of an $\ell$-mer $\alpha\in\mathcal{L}(W)$ is an $(\ell+1)$-mer $c\alpha\in\mathcal{L}(W)$, that has $\alpha$ as a suffix. Note that all left-extensions of a white $k$-mer are also white. 
    \item The \emph{left-contraction} of an $\ell$-mer $\alpha$ is the maximal interval of vertices that $\alpha[2..\ell]$ reaches. The left-contraction of a white $\ell$-mer may be black.
    \item For a vertex $v$, $E_\ell(v)$ is the array containing the number of distinct left-extensions of each $\ell$-mer reaching $v$, sorted by their colex order.
    \item For a vertex $v\in V(W)$, the intervals $\mathsf{left\_interval}_\ell(v)$ and $\mathsf{right\_interval}_\ell(v)$ are defined as:
    \begin{itemize}
        \item $\mathsf{left\_interval_{\ell}}(v)=  \{\text{maximal }[v_a, v_b]\subseteq V(W) : (\forall w\in[v_a, v_b])(\min^{\ell}_v\leadsto w)\}$.
        \item $\mathsf{right\_interval_{\ell}}(v)=  \{\text{maximal }[v_a, v_b]\subseteq V(W) : (\forall w\in[v_a, v_b])(\max^{\ell}_v\leadsto w)\}$
    \end{itemize}
\end{itemize}

\noindent We give a few more details on the intervals. If any of $\min^{\ell}_v$ or $\max^{\ell}_v$ does not exist, we define the respective left or right interval to be empty. Furthermore, for an interval of vertices $[v_a,v_b]$ where $v_a<v_b$, we note that $\max^{\ell}_{v_a}=\min^{\ell}_{v_b}$ by Corollary \ref{corollary:max_min}, and this $\ell$-mer is black.
%In this case, we call the left endpoint \emph{open} if $\min^{k-1}_{v_a}\neq \max^{k-1}_{v_a}$, otherwise \emph{closed}; similarly, we call the right endpoint \emph{open} if $\min^{k-1}_{v_b}\neq\max^{k-1}_{v_b}$, otherwise \emph{closed}. 
We represent the intervals by the two endpoints $v_a, v_b$. In the appendix we show how to compute the intervals for all $v$ and all $1\leq \ell\leq k$ in $O(nk)$ time, by scanning the vertices left-to-right and using the $\mathsf{share}_\ell$ function in each level.
\subparagraph*{Components.} The data structure is comprised of the following parts:
\begin{itemize}
    \item The standard Wheeler index for $W$ \cite{gagie2017wheeler}.
    \item For all $v\in V(W)$, the $(k-1)$-mer counts $C_{k-1}(v)$.
    \item For all $v\in V(W)$ and all $1\leq\ell\leq k$, with $pred(v)=(u_1,\ldots,u_d)$, a cumulative array $\overline{C_\ell(v)}[1..d]$ where $\overline{C_\ell(v)}[j]$ is the total number of $(\ell-1)$-mers reaching $u_1\ldots u_{j-1}$ that are smaller than $\min^{\ell-1}_{u_j}$.
    \item For all non-sink vertices $v$, a predecessor data structure over the array of cumulative sums of $E_{k-1}(v)$, named $\overline{E_{k-1}(v)}[1..C_{k-1}(v)]$. Given a number $j$ of distinct left-extensions, the query returns the smallest index $t$ such that $E_{k-1}(v)[1]+\ldots+ E_{k-1}(v)[t]\geq j$. We implement this data structure using constant-time rank and select operations \cite{JacobsonPhDRank, ClarkSelect}.
    \item For all $v\in V(W)$, the intervals $\mathsf{left\_interval}_{k-1}(v)$ and $\mathsf{right\_interval}_{k-1}(v)$, stored as described above.
\end{itemize}

Now, we describe how to implement the navigation queries with these data structures, and afterwards we analyse their space and time complexity.

\subsection{de Bruijn Graph Navigation}
The operations we implement return the \emph{handle} of a $k$-mer, listing outgoing labels from a $k$-mer, and performing forward traversal with an outgoing label.

\subparagraph*{Handle of a $k$-mer.} The handle of a $k$-mer $\alpha$ is a triple $(u, v, j)$, where $[u, v]$ is the maximal interval of $V(W)$ that $\alpha$ reaches, and $j$ is the rank of $\alpha$ in the colex order of $k$-mers reaching the left endpoint $u$. We obtain $u$ and $v$ via the standard forward-search on the Wheeler index, but we must keep track of the rank $j$. For $0\leq \ell\leq k,$ let $[u^\ell,v^\ell]$ be the interval of vertices that we obtain after matching $\alpha[1..\ell]$, starting from $[u^0,v^0]=[v_1,v_n]$. At any step where $u^\ell<v^\ell$, $j=C_\ell(u^\ell)$ because $\alpha[1..\ell]=\max^\ell_{u^\ell}$. If $u^\ell=v^\ell$ for some $\ell$, we update $j$ as follows. We find the rank $t$ of $u^{\ell-1}$ among the in-neighbours of $u^\ell$ with the Wheeler index, and we set $j\leftarrow \overline{C_\ell(u^\ell)}[t] + j$. %The operations with the Wheeler index take $O(\log \sigma)$ time, thus the whole procedure takes $O(k\log \sigma)$ time.

\subparagraph{Listing outgoing labels.} Given a handle $(u, v, j)$, with the standard Wheeler graph index we can output the distinct outgoing labels in $O(\log \sigma)$ time per label, using rank and select operations \cite{gagie2017wheeler}. 

\subparagraph{Forward traversal. }Given the handle $(u,v,j)$ of a $k$-mer $b\alpha$, $|\alpha|=k-1$, and a label $c$, we can return the handle of $\alpha c$ by first performing a left contraction on $b\alpha$, which returns an interval $[u'v'],$ and then matching $c$ from that interval with the Wheeler index and updating the rank (this is done in the same way as when we return the handle of a full $k$-mer).

To perform a left contraction we distinguish two cases. If $u<v$, then the left contraction is equal to $\mathsf{right\_interval}_{k-1}(u)$, and the rank is $C_{k-1}(u)$. If $u=v$, then the rank of $\alpha$ is equal to $select(E_{k-1}(u), j)$. If the new rank is $1$ or $C_{k-1}(u)$, then the interval is $\mathsf{left\_interval}_{k-1}(u)$ or $\mathsf{right\_interval}_{k-1}(u)$, respectively. Otherwise, $\alpha$ is a white $k$-mer, and the interval is $[u, u]$. %The operation takes $O(\log \sigma)$ time from the matching of the last character with the Wheeler index.

\subsection{Analysis.}
We store the $\ell$-mer counts $C_{\ell}(v)$ and the arrays $\overline{C_{\ell}(v)}$ as arrays of integers. The required integer size is given by the maximum number of $\ell$-mers, across all $\ell=1..k$. We denote by $n_\ell$ the number of $\ell$-mers of $W$, and with this notation, we obtain that we require $nk\log(\max_{1\leq\ell<k} n_\ell)$ bits to store the counts $C_{\ell}(v)$. The integers in $\overline{C_\ell(v)}$ grow up to $C_{\ell+1}(v)$, thus we require $2mk\log(\max_{2\leq\ell\leq k} n_\ell)$ bits to represent the arrays. The intervals take up $2nk(\log n + 2)$ bits. The Wheeler graph index is representable in $O(|W| + m\log\sigma+\sigma\log m)$ bits \cite{gagie2017wheeler}. 

Finally, to build the predecessor data structure on $\overline{E_{k-1}(v)}$, we encode $E_{k-1}(v)$ in unary, for example by $\bigodot_{1\leq j\leq |E_{k-1}(v)|}10^{E_{k-1}(v)[j]}$, where $\bigodot$ is the operator for concatenation of bit sequences. We build the classic constant rank and select data structures on it \cite{JacobsonPhDRank, ClarkSelect}. To analyse the space, observe that all left extensions counted by all arrays $E_{k-1}(v)$ for all $v$ represent distinct $k$-mers, except for at most $2n$ shared $k$-mers between the vertices. Therefore, the sum of all zeros in all $E_{k-1}(v)$ is bounded by $n_k+2n=|V(dBg)|+2n$. The number of ones in $E_{k-1}(v)$ corresponds to the number of $(\ell-1)$-mers arriving at $v$. Notice, however, that we only compute $E_{k-1}$ for non-sink vertices, thus for each distinct $(k-1)$-mer represented by a $1$-bit we have a distinct $k$-mer in $V(dBg)$ by right-extension. Thus, we get at most $|V(dBg)|+2n$ ones. The space becomes $O(n_k + |W|k\log(\max_{1\leq\ell\leq k} n_\ell) + \sigma\log m)$~bits\footnote{Here we used $\sigma\leq n$ \cite{gibney_complexity_2022}.}.

The time needed for the operations to obtain the handle of a $k$-mer are bounded by $O(k)$ Wheeler graph operations, which take $O(\log \sigma)$ each. Listing outgoing labels can be done in $O(\log\sigma)$ time per label. The forward traversal requires one Wheeler graph matching, and a predecessor on $\overline{E_{k-1}(v)}$, which can be implemented now with constant time rank and select on $E_{k-1}(v)$. Therefore, we arrive at the following theorem:

\begin{theorem}
    There exists a data structure for simulating navigation of the de Bruijn graph $dBg$ of the $k$-mers of a Wheeler graph $W$, which uses $O(n_k + |W|k\log(\max_{1\leq\ell\leq k} n_\ell) + \sigma\log m)$ bits of space. It implements retrieving the handle of a $k$-mer in $O(k\log\sigma)$ time, listing of outgoing labels in $O(\log\sigma)$ time per label, and forward traversal in $O(\log \sigma)$ time.
\end{theorem}

\subsection{Computation of the data structure}
We have discussed how to obtain the $k$-mer counts of $W$ in section \ref{sec:kmer_counting}, from which we obtain $C_{\ell}(v)$ and $\overline{C_\ell(v)}$, and we obtain the left and right intervals as shown in appendix \ref{appendix:b}. The Wheeler graph index is as in \cite{gagie2017wheeler}. Computing the arrays $E_{k-1}(v)$ is more challenging though, because its size is proportional to the number of $(k-1)$-mers of $W$, which can be exponential. We develop a recursive formula for $E_\ell$ in terms of $E_{\ell-1}$, which if it were implemented directly, it would lead to a $O((|W|+|dBg|)k)$ solution. Then, we show a smart dynamic program that implements it and computes $E_{k-1}(v)$ in $O(|W|k + n_k)$ time and space.

\subsubsection{Computation of $E_{k-1}(v)$}\label{subsection:computing_left_extensions}

Let $\alpha$ be an $\ell$-mer, and denote by $LE(\alpha, v)$ the number of left-extensions of $\alpha$ that reach $v$. First, we observe that all the $\ell$-mers reaching $v$ can be found by appending $\lambda(v)$ to the set of $(\ell-1)$-mers reaching the vertices of $pred(v)$ (formula \ref{eq:union_formula}). The same can be said about left extensions. We observe two facts:
\begin{enumerate}
    \item Let $\alpha$ be an $(\ell-1)$-mer reaching only one vertex $u\in pred(v)$, and assume it is the $j$th $(\ell-1)$-mer of $u$. Then,
    $
    LE(\alpha\lambda(v), v) = LE(\alpha, u) = E_{\ell-1}(u)[j].
    $
    \item Assume $\alpha$ reaches a maximal range of in-neighbours $[u_a,u_{a+1}\ldots,u_b]\subseteq pred(v)$, with $a<b$. Then,
    \begin{align}
    LE(\alpha\lambda(v), v)=LE(\alpha, u_a) + \sum_{a<j\leq b} (LE(\alpha, u_j) - \mathsf{share}_{\ell}(u_{j-1}, u_j)),\label{eq1}
    \end{align}
    where $LE(\alpha, u_a)=E_{\ell-1}(u_a)[C_{\ell-1}(u_a)]$ and $LE(\alpha, u_j)=E_{\ell-1}(u_j)[1]$. Note how we must avoid double-counting left-extensions of size $\ell$ that might reach several in-neighbours.
\end{enumerate}

The second case can only happen for $(\ell-1)$-mers that are the minimum or the maximum ones of an in-neighbour $u_j$, by Corollary \ref{corollary:max_min}. Similarly, a ``middle'' $(\ell-1)$-mer $\alpha$ that is not $\min^{\ell-1}_{u_j}$ or $\max^{\ell-1}_{u_j}$ will always fall into case 1, where $LE(\alpha\lambda(v), v)=LE(\alpha, u_j)$, that can be computed recursively. With these observations, we design an algorithm that computes $E_{k-1}(v)$ for all $v\in V(W)$ in two steps. First, we compute all $E_{\ell}(v)[1]$ and $E_{\ell}[C_\ell(v)]$ for all $v\in V(W)$ and $1\leq \ell\leq k$, which can be done in $O(|W|k)$ time, and can be used to compute the number of left extensions of $(k-1)$-mers in case 2. Then, we use a recursive algorithm to report the middle values of $E_{k-1}(v)$. To implement these algorithms, we will need all $\ell$-mer counts $C_{\ell}(v)$ for all $v\in V(W)$ and $1\leq\ell\leq k$, as well as the data structure to support $\mathsf{share}_\ell(u,v)$ queries and $\mathsf{ILCS}(v)$ queries.

\subparagraph*{Computing $E_{\ell}(v)[1]$ and $E_{\ell}(v)[C_{\ell}(v)]$. }First, we observe that $|E_1(v)|=1$, and $E_1(v)[1]=E_1(v)[C_1(v)]=[C_2(v)]$ by input-consistency. Then, for $\ell=2,3,\ldots,k-1$, we compute $E_\ell(v)[1]$ as follows. First, we scan the in-neighbours left-to-right until finding the first one with a non-zero count of $(\ell-1)$-mers. Let $u_s$ be this vertex. We add the count of left-extensions of the first $(\ell-1)$-mer of $u_s$, and in case that $(\ell-1)$-mer is shared by more in-neighbours (which can be checked with $\mathsf{share}_{\ell-1}(u_s,u_{s+1})$), updating the count as in equation (\ref{eq1}) above for every in-neighbour that shares the $(\ell-1)$-mer. The computation of $E_\ell(v)[C_{\ell}(v)]$ is similar, but going right-to-left, starting from the rightmost in-neighbour $u_e$ that has $C_{\ell-1}(u_e)>0$.
% We can express the quantity with the following formulas, where $pred(v)=(u_1,\ldots,u_d)$:
% \begin{gather*}
%     E_\ell(v)[1] = E_{\ell-1}(u_s)[1] + [\mathsf{ILCS}(u_s)\geq \ell-1]\cdot \sum_{s<j\leq d}\mathsf{share}_{\ell-1}(u_s,u_j)(E_{\ell-1}(u_j)[1] - \mathsf{share_\ell}(u_{j-1}, u_j)) \\
%     E_\ell(v)[C_\ell(v)] = E_{\ell-1}(u_e)[C_{\ell-1}(u_e)] + [\mathsf{ILCS}(u_e)\geq \ell-1]\cdot \\ \sum_{1\leq j<e}\mathsf{share}_{\ell}(u_j,u_e)(E_{\ell-1}(u_j)[C_{\ell-1}(u_j)] - \mathsf{share}_\ell(u_j,u_{j+1}))
% \end{gather*}
Clearly, we spend $O(1+\deg_{in}(v))$ time per vertex per value of $\ell$, thus in total we spend $O((n+m)k)=O(|W|k)$ time, and obtain a data structure of size $O(nk)$. 

\subparagraph*{Computing the rest of $E_{k-1}(v)$ recursively. }We sketch a recursive procedure called $\mathsf{left\_extensions}(w,\ell, E_{k-1}(v))$, that appends to one array $E_{k-1}(v)$ the counts of left-extensions that we find for $\ell$-mers of $w$. This procedure uses too many recursive calls, which we fix in the next section. 

We start by calling $\mathsf{left\_extensions}$ with $w=v$ and $\ell=k-1$. Let $pred(w)=\{u_1,\ldots,u_d\}$. We scan the in-neighbours of $w$ left-to-right, and we find the first one $u_s$ with non-zero count of $(\ell-1)$-mers, $C_{\ell-1}(u_s)>0$. We repeat the following steps for all in-neighbours starting from $u_j=u_s$: 
\begin{enumerate}
    \item Process the first $(\ell-1)$-mer $\alpha$ of $u_j$. If we have counted its number of left-extensions already, ignore it. Otherwise, obtain the count of its left-extensions from the previously computed $E_{\ell-1}(u_j)[1]$. If $\alpha$ is shared with other in-neighbours, we have to apply formula (\ref{eq1}) to count them. Append this count to $E_{k-1}(v)$. If we applied formula (\ref{eq1}), we start again at step 1 with the last in-neighbour that has $\alpha$ incoming. Otherwise, move to step 2. 
    \item If $C_{\ell-1}(u_j)>2$, recurse on $u_j$ with a call to $\mathsf{left\_extensions}(u_j,\ell-1,E_{k-1}(v))$.
    \item Process the last $(\ell-1)$-mer $\beta$ of $u_j$. If $\beta\lambda(w)$ is the last $\ell$-mer of $w$, stop. Otherwise, check $\mathsf{share}_{\ell-1}(u_j,u_{j+1})$. If it is false, we append $E_{\ell-1}(u_j)[C_{\ell-1}(u_j)]$ to $E_{k-1}(v)$. If it is true, we compute the number of left-extensions of $\beta$ with formula (\ref{eq1}), and we move to the last in-neighbour that has $\beta$ incoming.
\end{enumerate}

It is possible to compute these steps for all in-neighbours of $w$ in time $O(\deg_{in}(w))$ plus the time to make the recursive calls, by using the left and right intervals and the $\mathsf{share}$ function. We leave the details in appendix \ref{appendix:b}. Still, this procedure uses too many recursive calls, since for some $w$ and $\ell$, $\mathsf{left\_extensions}(w,\ell,E_{k-1}(v))$ may be called by many out-neighbours of $w$. However, we observe that every such call outputs the exact same numbers to $E_{k-1}(v)$. Therefore, using dynamic programming with memoization we can avoid computing repeated sub-problems.

\subparagraph*{Memoization. }We build the $(k-1)$-times unfolded graph $W_{k-1}$, and we add to all the vertices of $W_{\ell-1}$ a placeholder for two pointers $p$ and $q$. At the start of the algorithm, these pointers are empty. Crucially, we keep and reuse the pointers between computations of $E_{k-1}(v_i)$ for different $v_i$. We consider a call to $\mathsf{left\_extensions}(w,\ell, E_{k-1}(v_i))$ as a call on the vertex $w^\ell$ at layer $\ell$. Whenever we recurse on a vertex $w^\ell$ for the first time, we set $p$ to the next position in the current list $E_{k-1}(v_i)$, where the next value will be appended. When the recursion returns, $q$ is set to point to the last filled position in $E_{k-1}(v_i)$. The next time we recurse on $w^\ell$, whether in the computation of $E_{k-1}(v_i)$ or another $E_{k-1}(v_j)$, instead of running the algorithm, we output all values between $p$ and $q$ on $E_{k-1}(v_i)$. 

\subparagraph*{Complexity. }Clearly, now there are at most $O(nk)$ recursions during the computation of \emph{all} $E_{k-1}(v)$. These take time $O(1+\deg_{in}(v))$ once per vertex, when the pointers $p$ and $q$ have not been computed yet. This amounts to $O(|W|k)$ time in total. All other time is spent reporting values of $E_{k-1}(v)$, which as we argued before, there are $O(|V(dBg)|)=O(n_k)$ of them. We have finally showed the following theorem:

\begin{theorem}
    We can build the data structure for de Bruijn graph navigation on the $k$-mers of a Wheeler graph $W$, in time and space $O(n_k + |W|k\log(\max_{1\leq\ell\leq k} n_\ell)+\sigma\log m)$.  
\end{theorem}

\section{Conclusion}\label{sec:conclusion}

This work, to the best of our knowledge, is the first to study the problem of $k$-mer counting on graphs. We have introduced the problem and established both hardness results and tractable cases. While the problem turns out to be $\#$P-hard, our results show that the problem is in P in the class of Wheeler graphs, and by extension on any class of graphs that can be transformed into a Wheeler graph of polynomial size, without changing the $k$-mers. As the first application for our $k$-mer counting algorithm, we have shown how to simulate the de Bruijn graph of the $k$-mers of a Wheeler graph, without having to construct $k$-mers explicitly, improving on the state-of-the-art in this respect~\cite{siren2017indexing}.

The tractability of the problem on Wheeler graphs raises several natural open questions. First, can we find the smallest Wheeler graph that preserves the $k$-mers of a given input graph? Second, which kinds of graphs admit polynomial-size Wheelerizations with the same $k$-mer set? Third, what is the best achievable polynomial degree of $n$ for algorithms running in $O(\textrm{poly}(n) \log k)$ time on Wheeler graphs? This work thus marks a first step toward a broader theory of $k$-mer counting in graphs.

%As another application, our algorithms could be used to quantify the similarity of two labeled graphs via the \emph{Jaccard index} of the $k$-mer sets of the graphs. The Jaccard index of two sets $A$ and $B$ is the ratio $|A \cap B| / |A \cup B|$, which is computable with our algorithms.

\newpage

\bibliography{references}
\clearpage
\appendix
\section{Proof of Lemmas \ref{lemma:kmer_order} and \ref{lemma:lcs_and_kmers}}\label{appendix:lemma6}
\lemmashared*
\begin{proof}
    If $\alpha = \beta$, the claim is trivially true. Otherwise, let $i$ be the largest index such that $\alpha[i] \neq \beta[i]$.
    Let $\gamma = \alpha[i+1..k] = \beta[i+1..k]$. 
    There exists a path $u_1, \ldots u_\ell$ labeled with $\gamma$ such that $u_\ell = u$, and a path $v_1, \ldots v_\ell$ labeled with $\gamma$ such that $v_\ell = v$.  We also have $u_j \neq v_j$ for all $j = 1, \ldots, \ell$, or otherwise the graph is not deterministic, so we can iterate Lemma \ref{lemma:det_wheeler_graph} for pairs of edges $(u_j, u_{j+1}, c)$, $(v_j, v_{j+1}, c)$, starting from $u_\ell<v_\ell\Rightarrow u_{\ell-1}<v_{\ell-1}$ to obtain $u_1 < v_1$. Since $\lambda(u_1) = \alpha[i] \neq \beta[i] = \lambda(v_1)$, if $\alpha[i]\succ\beta[i]$ then W1) implies $u_1>v_1$, a contradiction. Thus, it must hold that $\alpha[i] \prec \beta[i]$, which implies $\alpha \prec \beta$.
\end{proof}

\lemmaorder*
\begin{proof}
    The lemma follows from two facts:
    \begin{enumerate}
        \item $u$ and $v$ can only share a $k$-mer if $|\sup_u|\geq k$ and $|\inf_v|\geq k$,
        \item $|\sup_u|\geq k$ implies that $\max^k_u$ is a suffix of $\sup_u$, and $|\inf_u|\geq k$ implies that $\min^k_u$ is a suffix of $\inf_u$.
    \end{enumerate}
    Thus, if $\mathsf{LCS}(\sup_u,\inf_v)\geq k$, then $|\sup_u|\geq k, |\inf_v|\geq k$, and 2) implies $\max^k_u=\min^k_v$. Otherwise, if $\mathsf{LCS}(\sup_u,\inf_v)<k$, either a) $|\sup_u|<k\vee|\inf_v|<k$, and 1) does not hold, so the $k$-mers cannot be shared; or $|\sup_u|\geq k \wedge |\inf_v|\geq k$, but then by 2) $\max^k_u$ and $\min^k_v$ are suffixes of $\sup_u$ and $\inf_v$, and then their \textsf{LCS} is less than $k$.
    
    For the proof, let $I_u$ be the set of left-infinite strings reaching $u$, union with the set of finite strings starting at a source, and ending at $u$. It is easy to check that the definitions of $\inf_u$ and $\sup_u$ in \cite{DBLP:conf/cpm/AlankoCCKMP24} are equivalent to $\inf_u=\min I_u$ and $\sup_u=\max I_u$. 
    Observe that if $\max^k_u$ exists, there exists a maximum string in $I_u$ suffixed by $\max^k_u$: Let $M$ be the walk, possibly left-infinite, that is created right-to-left by following i) first the edges that form $\max^k_u$, and then ii) the maximum incoming edge to the last found vertex on the walk. This walk ends at $u$ and is either a left-infinite string, or starts in a source; thus it spells a string in $I_u$, that is suffixed by $\max^k_u$. The same argument holds for a minimum string in $I_u$ suffixed by $\min^k_u$, following the smallest incoming edge.

    We prove 2) first. Assuming that $|\sup_u|\geq k$, let $\alpha$ be the suffix of length $k$ of $\sup_u$. We have that $\max^k_u\succeq \alpha$, because $\max^k_u$ can always be chosen as $\alpha$. By contradiction, assume that $\max^k_u\succ \alpha$. Now, consider the maximum string $\gamma\in I_u$ that is suffixed by $\max^k_u$: Since $\max^k_u\succ \alpha$, then $\gamma\succ\sup_u$, but this contradicts $\sup_u=\max I_u$. Similarly, if $|\inf_u|\geq k$ and $\alpha$ is the suffix of length $k$ of $\inf_u$, we have $\min^k_u\preceq\alpha$ because we can choose $\min^k_u=\alpha$. But if $\min^k_u\prec\alpha$, then the minimum string $\gamma\in I_u$ suffixed by $\min^k_u$ is strictly smaller than $\inf_u$, a contradiction.

    To prove 1), we show two intermediate claims:
    \begin{enumerate}[A)]
        \item If $\min^k_u$ exists, $|\inf_u|<k$ implies $\inf_u\prec\min^k_u$. To see why, let $\gamma$ be the minimum string in $I_u$ suffixed by $\min^k_u$. We have $\gamma\neq\inf_u$ because they are of different lengths, but $\inf_u=\min I_u$, therefore $\inf_u\prec\gamma$. Since $\min^k_u$ is a suffix of $\gamma$ of length $k$, but $|\inf_u|<k$ and $\inf_u\prec\gamma$, then we have $\inf_u\prec\min^k_u$.
        \item $\max^k_u\preceq\sup_u$ regardless of the length of $\sup_u$. Let $\gamma$ be the maximum string in $I_u$ suffixed by $\max^k_u$: Since $\max^k_u$ is a suffix of $\gamma$, $\max^k_u\preceq\gamma$, and since $\sup_u=\max I_u$, we have $\gamma\preceq\sup_u$, which implies B).
    \end{enumerate}
    Now, we have all the ingredients to prove 1): if $|\inf_v|<k$, we have by A) and B) that $\max^k_u\preceq\sup_u\preceq\inf_v\prec\min^k_v$, so the $k$-mers cannot be shared. So we can assume $|\inf_v|\geq k$, and by 2) $\min^k_v$ is a suffix of $\inf_v$. If $|\sup_u|<k$, because of the length of $\inf_v$, we have $\sup_u\prec\inf_v$, and since $\min^k_v$ is a suffix of $\inf_v$, we have $\sup_u\prec\min^k_v$, that by B) implies $\max^k_u\prec\min^k_v$. So, the only time when $u$ and $v$ can share a $k$-mer is when $|\sup_u|\geq k$ and $|\inf_v|\geq k$. $\qed$
\end{proof}

\section{Details of the $O(n^4\log k)$ algorithm}\label{appendix:n4logk}
We show here how to compute $T^W_\ell$ and $R([u_i, u_j], \beta)$. We use the $\mathsf{share}_\ell$ and $\mathsf{ILCS}_{\geq\ell}$ functions, and the algorithm by Kim et al \cite{kim_et_al:LIPIcs.CPM.2023.16}. Given that we can compute the $\mathsf{LCS}$ array in time $O(m+n^2)$ \cite{DBLP:conf/cpm/AlankoCCKMP24}, the $\mathsf{share}_\ell$ and $\mathsf{ILCS}_{\geq\ell}$ functions can be obtained directly from it and an RMQ data structure, as shown before (formula \ref{equation:share}).
The algorithm of Kim et al. constructs the sets $P_\ell(v)$ at each iteration, that contain the vertices $u$ such that there exists a path spelling $suf_\ell(\inf_v)$ from $u$ to $v$. We make a matrix $D^{\inf}_\ell$ of size $n\times n$, where $D^{\inf}_\ell(u,v)=1$ if $u\in P_\ell(v)$. Furthermore, we index the columns of $D^{\inf}_\ell$ for range maximum queries \cite{Fischer-RMQ}. With this data structure, we can answer the following query: For an interval $[u_i,u_j]$, define $D^{\inf}_\ell([u_i,u_j],v)=1$ if there exists any $w\in[u_i,u_j]$ such that $suf_{\ell}(\inf_v)$ starts at $w$ and ends at $v$, otherwise $0$. We can build the corresponding data structures for the suprema, to answer the following queries: $D^{\max}_\ell([u_i,u_j],v)=1$ if and only if $suf_\ell(\sup_v)$ starts at some $w\in [u_i,u_j]$ and ends at $v$. Now, we can compute $T^W_\ell([u_i,u_j],v)$ with the following formula:
\begin{align*}
    T_\ell^W([u_i,u_j], v)&=T_\ell([u_i,u_j],v) - \mathsf{share}_{\ell}(v,v-1)\cdot D_{\ell}^{\inf}([u_i,u_j],v) \\
    & -\mathsf{share}_\ell(v,v+1)\cdot D_{\ell}^{\sup}([u_i,u_j],v])\cdot [\mathsf{ILCS}(v)<\ell].
\end{align*}

It takes time $O(n^2)$ to compute $P_\ell(v)$ for all $v$ and for both infima and suprema in each iteration. Then, we spend $O(n^2)$ time to build the matrices $D^{\inf}_\ell$ and $D^{\sup}_\ell$ and their correspondent range maximum data structures. The queries take $O(1)$ time. Thus, after $O(n^2)$ time, we can compute $T^W_{\ell}([u_i,u_j],v)$ in constant time from $T_{\ell}([u_i,u_j],v)$. 

To compute $R([u_i,u_j],\beta)$ for $\beta\in B_{\ell/2}$, we first compute all intervals $R([u_1,u_n],\beta)$ by scanning the vertices left-to-right and looking for the maximal intervals of vertices that share an $\ell$-mer, aided by the functions $\mathsf{share}_{\ell}$ and $\mathsf{ILCS}$ (see a more detailed version in appendix \ref{appendix:b}). Once we have them, we restrict all found intervals $[w_a,w_b]=R([u_1,u_n],\beta)$ to the vertices reachable from $[u_i,u_j]$ as follows. Scanning $[w_a,w_b]$ left to right, we change the left endpoint of $[w_a,w_b]$ to the first vertex $w_p\in [w_a,w_b]$ such that $D^{\sup}_\ell([u_i,u_j],w_p)=1$. Similarly, we change the last endpoint to the last vertex $w_q$ such that $D^{\inf}_{\ell}([u_i,u_j],w_q)=1$. Since there are $O(n)$ such intervals, and their intersections are of size $1$, this procedure can be done in $O(n)$ time for a given interval $[u_i,u_j]$. Across all $O(n^2)$ intervals, we have a total of $O(n^3)$ work. 

\section{Computing intervals and left-extensions}\label{appendix:b}
First, we show how to compute all $\mathsf{left\_interval}_\ell(v)$ and $\mathsf{right\_interval}_\ell(v)$ for all vertices $v$ and all $1\leq \ell\leq k$ in $O(nk)$ time once we have the $\ell$-mer counts. Then, we give the missing details for the left-extension counting algorithm of subsection \ref{subsection:computing_left_extensions}.

\subparagraph*{Computing intervals.}For any $\ell$, the sources $s$ have $\mathsf{left\_interval}_\ell(s)=\mathsf{right\_interval}_\ell(s)=\emptyset$. For $\ell=1$, the left and right intervals of any vertex are the same, and they correspond to the intervals of vertices with the same incoming label. For $\ell>1$, first we find for every non-source $v$ if $\mathsf{left\_interval}_\ell(v)=\mathsf{right\_interval}_\ell(v)$, by checking if $\min^\ell_{v}=\max^{\ell}_{v}$. This is easy: $\min^{\ell}_v=\max^{\ell}_v$ if and only if $C_{\ell}(v)=1$.
%Let $pred(v)=u_1,\ldots,u_d$, and find the first $u_s\in pred(v)$ and the last $u_e\in pred(v)$ such that $C_{\ell-1}(u_s)>0$ and $C_{\ell-1}(u_e)>0$, respectively. Then, if for all $s\leq p \leq e$ it holds that $\mathsf{left\_interval}_{\ell-1}(u_p)=\mathsf{right\_interval}_{\ell-1}(u_e)$, and for all $s\leq p<e$ it holds that $\mathsf{share}_{\ell-1}(u_p,u_{p+1})=1$, we have $\min^{\ell}_{v}=\max^{\ell}_{v}$; otherwise they are different. 
%For $\ell>1$, for every vertex $v_i$ we can check if its left or right interval are empty, by checking the interval at level $\ell-1$ of the leftmost and rightmost in-neighbour, respectively. If they are empty, the corresponding interval at level $\ell$ is empty too.

Once we have this information, it is easy to run a loop through the non-source vertices, from smallest to largest, to find the maximal runs of vertices $v_i..v_j$ such that for all $i\leq p<j$ we have $\mathsf{share}_{\ell-1}(v_p,v_{p+1})=1$, and for all $i<p<j$ it holds that $\mathsf{left\_interval}_{\ell}(v_p)=\mathsf{right\_interval}_{\ell}(v_p)$. With this information, we can find all maximal intervals of all the vertices and assign them correctly. The algorithm takes $O(1)$ time for each vertex in each level, thus it takes $O(nk)$ time in total.

\subparagraph*{Computing left-extensions. }We left out some details on the recursive algorithm to compute left-extensions. For a vertex $w^\ell$ and an in-neighbour $u_j$ of $w$, let $\alpha$ and $\beta$ be the first and last $(\ell-1)$-mers of $u_j$, respectively. We clarify the following points:
\begin{itemize}
    \item How to compute formula \ref{eq1} in time proportional to the number of in-neighbours that the $(\ell-1)$-mer reaches. Note that this may be smaller than the size of the interval of that $(\ell-1)$-mer.
    \item How to keep track of whether we have counted the number of left-extensions of $\alpha$.
    \item How to know if $\beta\lambda(w)$ is the last $\ell$-mer of $v$.
\end{itemize}

For the first point, when reporting the number of left-extensions of the last $(\ell-1)$-mer of $u_j$, it may happen that $\mathsf{right\_interval}_{\ell-1}(u_j)$ is large, but it contains many vertices that are not in-neighbours of $w$. To compute formula \ref{eq1} efficiently, we get the right endpoint $w_r$ of $\mathsf{right\_interval}_{\ell-1}(u_j)$, and we iterate the in-neighbours of $w$ starting from $u_j$ as long as they are smaller than $w_r$. For each of them, we add the appropriate term to the sum, and we report the value at the end.

For the second point, we keep a variable $\mathsf{report\_min}$ throughout the algorithm, that indicates whether we should report the number of left-extensions of $\min^{\ell-1}_{u_j}$ or not. At the start of the algorithm it is \emph{false}, since the number of left-extensions of $\min^{\ell}_w$ has been counted already (in $E_\ell(w)[1]$). When processing the last $(\ell-1)$-mer $\beta$ of an in-neighbour $u_j$, it may or may not be shared with more in-neighbours. If it is not shared, then $\mathsf{report\_min}$ is set to $true$. Otherwise, if it is shared, and the last in-neighbour $u_b$ of $w$ that $\beta$ reaches has $C_{\ell-1}(u_b)>1$, then we set $\mathsf{report\_min}$ to $false$, otherwise to $true$. 

For the last point, to avoid reporting the number of left-extensions $\max^{\ell}_w$, we have two cases. If we arrive at the last in-neighbour of $w$, we just avoid running the last step of the algorithm. The other case is if, for some $j<d$, it happens that $\mathsf{right\_interval}_{\ell-1}(u_j)$ contains $u_d$, and $C_{\ell-1}(u_d)=1$. Thus, we must check every time we compute formula \ref{eq1} if this is the case, and stop if so.

\end{document}